\documentclass[a4paper,11pt]{article}
\usepackage[utf8]{inputenc}
\usepackage[T1]{fontenc}
\usepackage[margin=3cm]{geometry}

\usepackage{courier}
\usepackage{amsmath,amsfonts,amssymb,amsthm}
\usepackage{graphicx}
\usepackage{natbib}
\usepackage{color}
\usepackage{comment}

\usepackage{latexsym}
\usepackage{bbm}
\usepackage{natbib}
\usepackage{verbatim}

\newcommand{\1}{{\mathbbm 1}}

\newtheorem{Prop}{Proposition}[section]

\begin{document}

\title{Conditional bias robust estimation of the total of curve data  by sampling in a finite population: an illustration on electricity load curves}

\author{Herv\'e Cardot$^{(1)}$,  Anne De Moliner$^{(1,2)}$ and Camelia Goga$^{(3)}$ \\
$^{(1)}$ Institut de Math\'ematiques de Bourgogne, UMR 5584 CNRS \\
Universit\'e de Bourgogne Franche-Comt\'e, France\\
$^{(2)}$ EDF LAB, Palaiseau, France\\
$^{(3)}$ Laboratoire de Math\'ematiques de Besan\c{c}on, UMR 6623 CNRS \\
Universit\'e de Bourgogne Franche-Comt\'e, France\\
 {\small \texttt{herve.cardot@u-bourgogne.fr},  \texttt{camelia.goga@univ-fcomte.fr}} 
}

\maketitle
  
\begin{abstract}
For marketing or power grid management purposes, many studies based on the analysis of the total electricity consumption curves of groups of customers are now carried out  by electricity companies. 
Aggregated total or mean load curves are estimated using individual curves measured at fine time grid and collected according to some sampling design. Due to the skewness of the distribution of electricity consumptions, these samples often contain outlying curves which may have an important impact on the usual estimation procedures. We introduce several robust estimators of the total consumption curve which are not sensitive to such outlying curves. These estimators are based on the conditional bias approach and robust functional methods. We  also derive mean square error estimators of these robust estimators and finally, we evaluate and compare  the performance of the suggested estimators  on Irish electricity data. 
\end{abstract}


\section{Introduction and context}

Many studies carried out  by electricity companies are based on the analysis of total electricity consumption curves measured at fine time scales (often half-hourly) for one or several groups of clients sharing some common characteristics (e.g. customers from the same electricity provider, having a particular electric equipment or living in a given geographic area). The aim of these studies can be for example to assist the power grid manager in maintaining the balance between electricity consumption and production at every instant on the power grid. The total consumption curves can also be used to help the Sales Division to quantify the impact of a specific electric use or equipment on the electricity consumption, to build new innovative pricing strategies or to create new services based on customers consumption analysis.

In order to avoid technical and budgetary constraints due to limited bandpass or storage cost of huge databases, or in order to preserve privacy, the strategy of selecting a sample of individual curves from the whole datasets is often employed. The total consumption curve or the load curve of each population of interest is then estimated by using the curves of the customers belonging to the sample.
The estimation with survey sampling techniques of parameters of interest  such as  the total or the mean, the median or the principal components when the data are curves  has been developed over the last years: \cite{cardot2010properties}, \cite{cardot2011horvitz}, \cite{cardot2013confidence} and \cite{chaouch_goga_2012}.  Several sampling designs and estimators have been compared by means of simulation on real electricity data set in \cite{cardot2013comparison} and some asymptotic properties have been established in \cite{cardot2013uniform} and \cite{cardot_goga_lardin_scandin}.  We cite also \cite{degras_20} for the Horvitz-Thompson estimation with optimal rotation of samples. A recent review of research works  in this area is given in \cite{lardin_cardot_goga_2014}.

\medskip

We address here the estimation of the total consumption curve in presence of outlying curves. Following \cite{chambers1986outlier}, we consider only representative outlying curves, namely curves which are representative for some non-sampled units and that do not come from measurement errors.
With electricity data, it is not unusual to have units with consumption electricity much higher than the rest of the population (see Figure~\ref{atrandom}). Such outlying curves may have a huge impact on the estimation and it is very important to detect and treat them correctly.  In order to detect such outlying curves, we use the notion of  depth of a curve introduced by \cite{lopez2009concept}. 

In a finite population setting, stratification is a good method to reduce the potential impact of outlying curves. More exactly, the population is divided into disjointed subpopulations called strata and units from the same stratum are as similar as possible according to several criteria. Unfortunately, due to wrong classifications or sudden changes some units may be very different from the other units belonging to the same stratum. These units are influential and deteriorate the stratum homogeneity and the variance of the usual estimators for the total or the mean will be large. More generally,  a unit is considered influential if, in a given configuration: study population and variable, parameter, estimator and sampling design, its value has a great impact on the variance of the estimator (\cite{Favre-Martinoz_these}). 

Several robust estimators not sensitive to influential units have been suggested in the survey sampling setting for real data,  that are not curves. We can cite for example \cite{chambers1986outlier}, \cite{gwet1992outlier}, \cite{rivest1994statistical}, \cite{kokic1994optimal}, \cite{welsh_ronchetti}. Broadly speaking, these estimators are based on winsorization techniques which consist in down-weighting the influence of outlying units. This is performed  by considering a thresholding function depending on a tuning constant whose value must be chosen carefully. 
The reader is referred to Chapter 11  of \cite{handbookA_2009} for a detailed presentation of the main methods dealing with outliers in survey data. Recently, \cite{beaumont2013unified} considered a new robust estimator for finite population totals. This new approach is based on the notion of conditional bias introduced by \cite{munoz1995new} to measure the influence of a unit and  is closely related to the estimator of \cite{chambers1986outlier}.  Besides, the conditional bias approach does not require to introduce a superpopulation model. Another popular approach for building robust estimators for survey data is the one suggested by  \cite{kokic1994optimal}. 
The use of Kokic and Bell's method would require the knowledge of a model for the probability distribution for functional data. Such superpopulation models  are generally very complex in our curve data framework and cannot generally be reduced to parametric models with a small number of parameters.  A recent comparison of robust estimation strategies in a finite population by \cite{DM2018} has also  shown that a misspecification of the superpopulation model can deteriorate much the accuracy of the robust estimator based on Kokic and Bell's approach.  For these two reasons, we did not consider further the  Kokic and Bell's approach in the present work.

\medskip 

The aim of this paper is to build robust design-based estimators of the total consumption curves which are less sensitive to influential curves. Since generally the curve data are observed at a finite number of time instants, the easiest and most intuitive way to construct such a robust estimator is to apply the method suggested by \cite{beaumont2013unified} at each instant of time. Unfortunately, this method does not take into account possible temporal correlations. In order to deal with this issue, we can transform the data by using dimension reduction methods such as functional principal component analysis or projection of the data onto basis functions. 
We suggest in this paper to perform a robust principal component analysis (PCA) as introduced by \cite{locantore1999robust}  in order to obtain uncorrelated real principal components. The total consumption curve may be then approximated in a smaller dimensional space spanned by robust eigenfunctions.
 Then, the coordinates in this new robust basis can be robustified by using  the method of \cite{beaumont2013unified} and a second robust estimator for the total consumption  is then obtained. 
 
 Instead of using robust PCA, one may also project the data onto a basis functions, such as wavelets, which are known to be effective to deal with irregular temporal signals such as individual electricity load curves (see {\it e.g.} \cite{Mallat98}). A third robust  estimator for the total consumption  is then obtained by robustifying the coordinates in the wavelet  basis.

The choice of a positive cut-off constant $c$ is required to build these estimators. Choosing an adequate value  is crucial since a trade-off between bias and variance must be made. We suggest in this paper a new criterion for choosing this tuning constant based on the $q$th power of the conditional bias. We also introduce a functional  truncation method based on the concept of depth of curves (\cite{lopez2009concept}) as a functional measure of outlyingness. This method consists in finding a zone which entirely  contains the conditional biases considered as ``inliers''  and to use the upper and lower bounds of this zone as truncation limits. A fourth estimator may then be constructed. 

\medskip

This paper is organized as follows: in Section 2, we describe the estimation of totals with sampling designs from a finite population of curves and we extend the definition of the conditional bias for functional data. In Section 3, we apply point-wisely   the approach of \cite{beaumont2013unified} for building a robust estimator for the total curve and we use their minimax criterion for choosing the tuning constant as well as a new one based on the $q$th power of the conditional bias.
We introduce in Section 4 two robust estimators based on dimension reduction techniques  and in Section 5, a robust estimator built for the functional truncation method based on the modified band depth as suggested by \cite{lopez2009concept}. In Section 6, we address the question of the estimation of pointwise mean square error.  Due to confidentiality reasons the electricity data from EDF can not be used for publication. We illustrate, in Section 7, the performances of the different robust  approaches on the estimation of the total curve on Irish electricity consumption curves.  Concluding remarks are given in brief  Section 8 and some proofs are postponed in an Appendix.

\begin{figure}[ht]
\begin{center}
\includegraphics[scale=0.7]{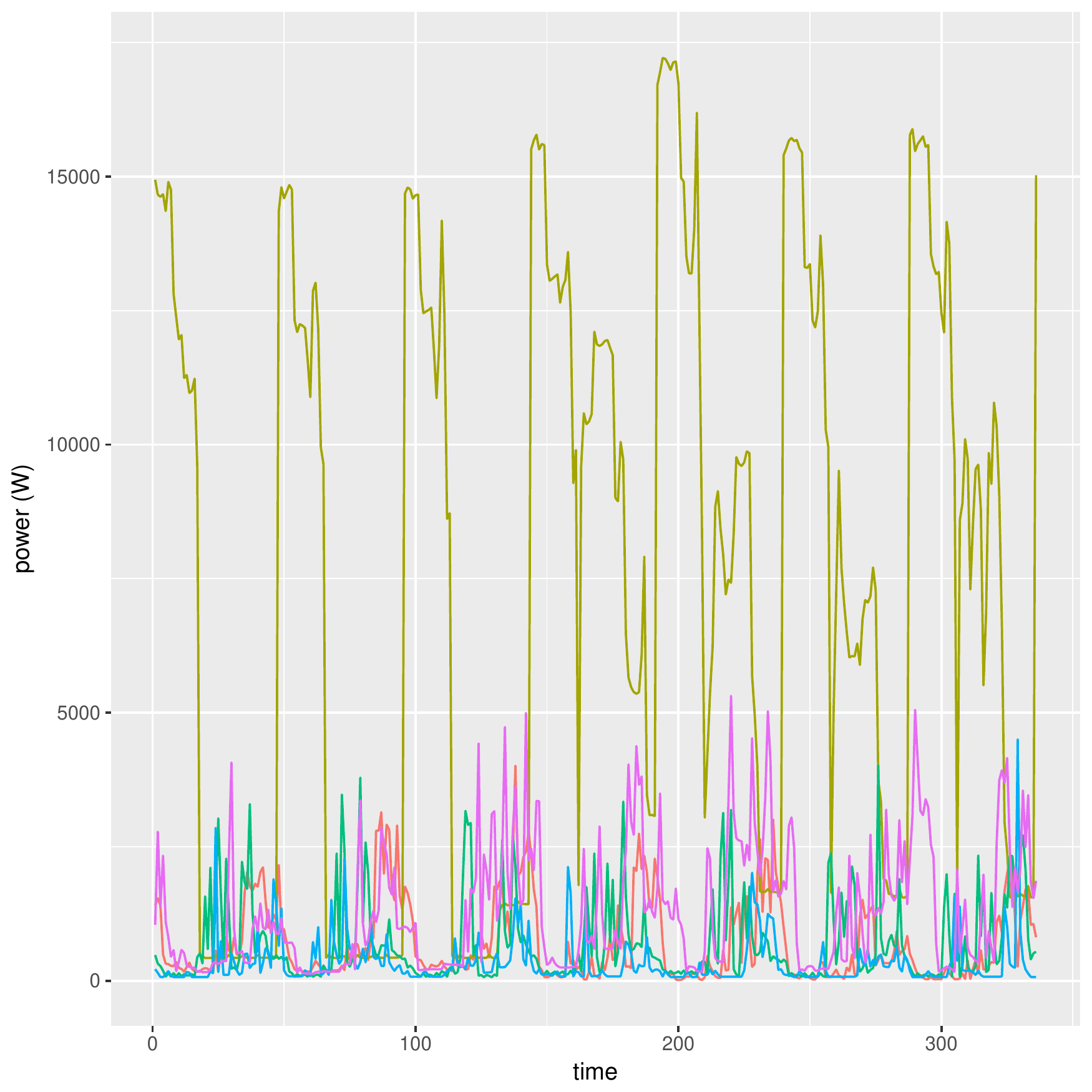} 
\caption{A sample of five load electricity curves measured every half an hour over a period of one week.}
\label{atrandom}
\end{center}
\end{figure}

\section{Robust estimation in a finite population of curves}
\subsection{Notations and framework}

Let $U$ be a population of interest of known size $N$. To each unit $i$ of the population we associate a (load) curve defined over a time interval $[0,T]$: for each unit $i$ we have a function of time $Y_i(t), t \in [0,T]$, where the continuous index $t$ represents time.\\ 
Our goal is to estimate the total curve $t_Y$ over the population:
\begin{equation}
t_Y=\sum_{i \in U} Y_i, 
\end{equation}
with value $t_Y(t)=\sum_{i \in U} Y_i(t)$ for each instant $t \in [0,T].$ In practice, the curves are not observed continuously for $t \in [0,T]$ but only for a set  of $D$ measurement instants $0=t_1 < t_2 < ... < t_D=T$ which are generally supposed to be equi-spaced and the same for all units. Under weak assumptions on the number of discretization points, the regularity of the trajectories and the sampling design, it can be shown that the approximation error due to linear interpolation or kernel smoothing is negligible compared to the sampling error (see \cite{cardot2011horvitz} and \cite{cardot2013confidence}).

To evaluate and compare the different approaches, we consider in this work a test population composed of $N=3994$ electricity consumption curves extracted from the Irish Commission for Energy Regulation (CER) Smart Metering Project that was conducted in 2009-2010 (CER, 2011)\footnote{The data are available on request at the address: \\ \texttt{http://www.ucd.ie/issda/data/commissionforenergyregulation/}}. The electricity consumptions are recorded during one week, from the 18th to the 24th of January 2010,  we have $D=336$ points in time (see Section \ref{section_application} for more details).  We display in Figure~\ref{atrandom}, the electricity consumption curves for five smart meters selected from that population. 

A sample $s$ of size $n$ is selected from $U$ according to a random sampling design $p(\cdot)$. We denote by $I_i$ the sample membership indicator of unit $i$ which is equal to 1 if the unit $i$ belongs to the sample $s$ and zero otherwise. The probability that unit  $i$ will be included in a sample is denoted  by $\pi_i= P(I_i=1)=\sum_{s, i\in s}p(s)$ and the probability that both of the units $i$ and $j$ will be included is denoted by  $\pi_{ij}= P(I_i I_j =1)=\sum_{s, (i,j)\in s}p(s). $ The first-order inclusion probabilities $\pi_i$ and the second-order inclusion probabilities $\pi_{ij}$ are assumed to be known and strictly positive.
We also assume that $\pi_i$ and $\pi_{ij}$ do not depend on time $t$.

We will particularly be interested by two simple sampling designs,  simple random sampling without replacement (SRS) and stratified sampling with simple random sampling within strata (STR). In STR, units with similar characteristics (according to some auxiliary information) are grouped into disjointed strata $U_h$ of size $N_h$ for $h=1, \ldots, H$. A simple random sampling without replacement  $s_h$ of size $n_h$ is selected from $U_h$ and the selection in one stratum is independent of the selection in all other strata. 
Note also that in the following, inference is made under the design-based approach in a finite population setting. This means that the sample membership indicators $I_i, \ i \in U$ are  binary random variables and the values of the variable of interest $Y_i$ are treated as being deterministic. In this context, the total curve $t_Y$ can be estimated by  the Horvitz-Thompson estimator,
\begin{align}
\hat{t}_Y(t) &=\sum_{i \in s} d_i Y_i(t), \quad t\in [0,T],
\label{HTeq}
\end{align}
where $d_i=1/\pi_i$, $i \in U$ are the sampling weights. The  Horvitz-Thompson estimator $\hat t_Y$ is a random curve,  with covariance function  given by 
\[
Cov(\hat{t}_Y(r),\hat{t}_Y(t))=\sum_{i\in U}\sum_{j\in U}(\pi_{ij}-\pi_i\pi_j)\frac{Y_i(r)}{\pi_i}\frac{Y_j(t)}{\pi_j}, \quad \mbox{for all } r, t \in [0,T].
\]
A unit $i$ with a large sampling weight $d_i$ and a large value of $Y_i$ at some time instant $t$ is influent for the Horvitz-Thompson estimator given in  \eqref{HTeq} since it increases considerably the covariance of the Horvitz-Thompson estimator given above.\\

\subsection{Conditional bias when the data are curves}\label{direct}
In order to construct robust estimators, \cite{beaumont2013unified} have used the conditional bias as a tool for quantifying  the influence of sampled and non sampled units on an estimator. The conditional bias, as defined by \cite{beaumont2013unified}  is, in a design-based approach, the expectation of the estimator conditionally to the inclusion indicator $I_i$ of the unit $i$. In our context, the  conditional bias of a sampled unit is a function of time $t$,
\begin{equation}
B_{1i}^{HT}(t)=E_p (\hat{t}_Y(t) |I_i=1) -t_Y(t)=\sum_{j\in U}\left(\frac{\pi_{ij}}{\pi_i\pi_j}-1\right)Y_j(t), \quad t\in [0, T], \label{def_biais_cond}
\end{equation}
and for a non-sampled unit: 
\begin{equation}
B_{0i}^{HT}(t)=E_p (\hat{t}_Y(t) |I_i=0) -t_Y(t)=-\frac{1}{d_i-1}B_{1i}^{HT}(t), \quad t\in [0, T],
\end{equation}
where $E_p$ is the expectation with respect to the sampling design $p$. 
For  simple random sampling without replacement (SRS), the conditional bias have the following expression,
\begin{align*}
B^{HT}_{1i}(t) &=\frac{N}{N -1}\left(\frac{N}{n} -1\right)(Y_i(t) - \overline{Y}_U(t)), \quad i\in U, \quad t\in [0,T],
\end{align*} 
where $\overline{Y}_U(t)=\sum_{i\in U} Y_i(t)/N,$ and for stratified sampling with SRS within each stratum (STR), the conditional bias of a sampled unit $i$ belonging to the stratum $U_h$ is 
\begin{equation*}
B^{HT}_{1i}(t)=\frac{N_h}{N_h -1}\left(\frac{N_h}{n_h} -1\right)(Y_i(t) - \overline{Y}_{U_h}(t)), \quad i\in U_h, \quad t\in [0,T],\label{bias}
\end{equation*}
where $\overline{Y}_{U_h}(t)=\sum_{i\in U_h}Y_i(t)/N_h$ is the mean curve within stratum $h$. We can  see that for stratified sampling, a unit $i\in U_h$ has a large  influence if its value $Y_i(t)$ is far from the mean stratum $\overline{Y}_{U_h}(t)$ and its influence is even larger if it is associated with a large sampling weight $N_h/n_h$. %This happens for example for strata jumpers. 

We can see from (\ref{def_biais_cond}) that the conditional bias $B^{HT}_{1i}(t)$ is unknown and must be estimated. A conditionally design-unbiased estimator of $B^{HT}_{1i}(t)$, given $I_i=1$, is:
\begin{eqnarray}
\hat B^{HT}_{1i}(t)=\sum_{j\in s}\left(\frac{\pi_{ij}-\pi_i\pi_j}{\pi_j\pi_{ij}}\right)Y_j(t), \quad \mbox{for all} \quad t\in [0,T].
\end{eqnarray}
In the case of SRS sampling, the conditional bias can be estimated by 
\[
\hat{B}^{HT}_{1i}(t)=\frac{n}{n -1}\left(\frac{N}{n} -1\right)(Y_i(t) - \overline{Y}_s(t)), \quad i\in U, \quad t\in [0,T],
\]
where $\overline{Y}_s(t)=\sum_{i\in s} Y_i(t)/n$ and for STR sampling,  it can be estimated by 
\begin{equation}
\hat{B}^{HT}_{1i}(t)=\frac{n_h}{n_h -1}\left(\frac{N_h}{n_h} -1\right) (Y_i(t) - \overline{Y}_{s_h}), \quad i\in U_h, \quad t\in [0,T],
\label{htex}
\end{equation}
where $\overline{Y}_{s_h}=\sum_{i\in s_h}Y_i/n_h$ is the sample mean of $Y$-values within the stratum $h. $

Consider again the test population of Irish electricity consumption curves. Two estimated conditional bias curves, with  simple random sampling of size $n=200$, are drawn in Figure~\ref{sanscorrection}. We can remark  on this small example how different, in shape and values, the conditional bias can be from one individual to another and also, from one instant of time to another.  
\begin{figure}[ht]
\begin{center}
\includegraphics[scale=0.7]{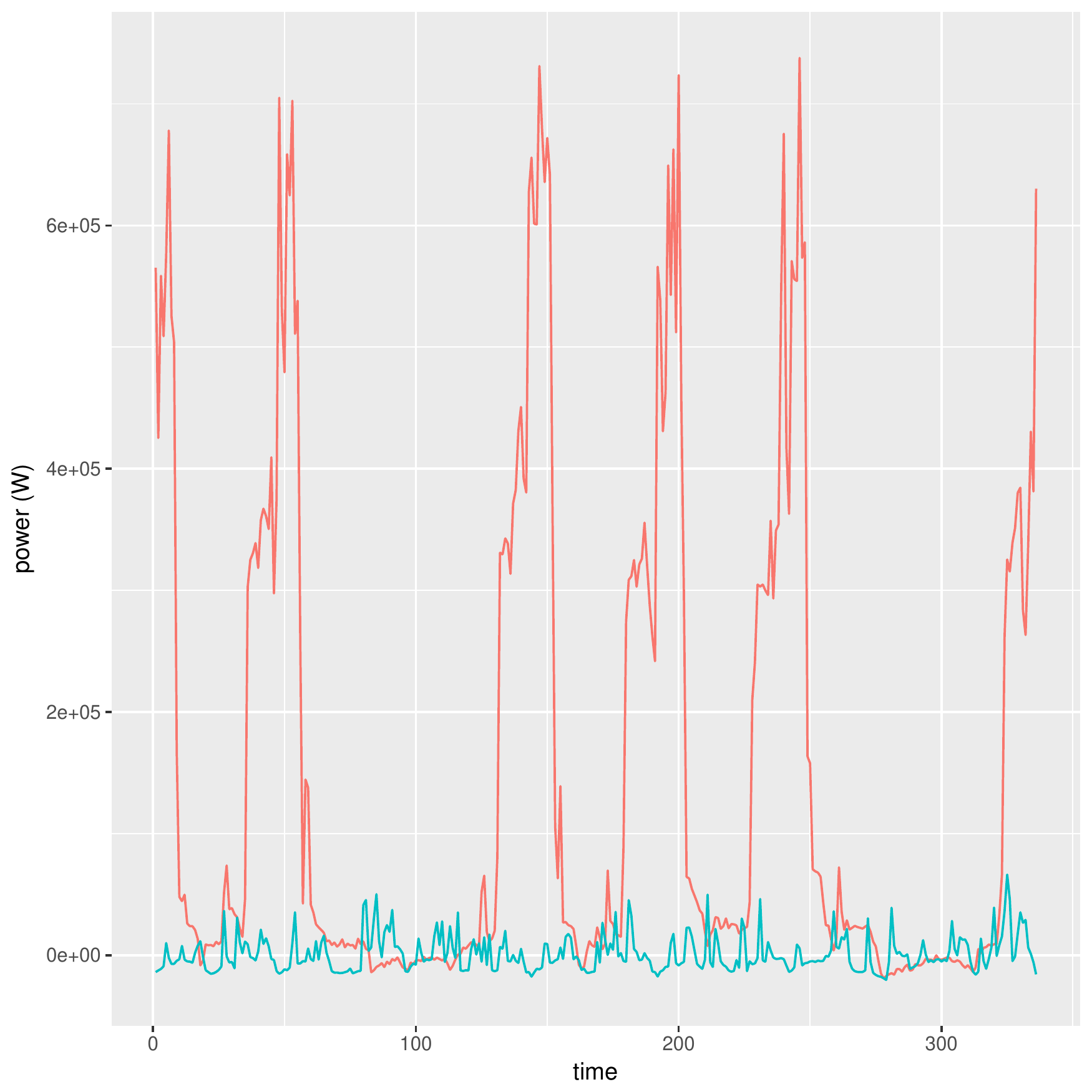} 
\caption{Two conditional bias curves estimated with simple random sampling of size $n=200$.}
\label{sanscorrection}
\end{center}
\end{figure}

Following the lines of \cite{beaumont2013unified}, we obtain in our functional setting:
\begin{equation}
\hat{t}_Y(t)-t_Y(t) = \sum_{i \in s} B^{HT}_{1i}(t) +  \sum_{i \in U-s } B^{HT}_{0i}(t) +\left(\sum_{i\in s}d_iA_i(t)-\sum_{i\in U}A_i(t)\right), \quad t\in [0, T],\label{sampling_error}
\end{equation}
where 
\begin{align*}
A_i(t) &=\frac{-1}{1-\pi_i}\sum_{j\in U, j\neq i}\frac{\pi_{ij}-\pi_i\pi_j}{\pi_j}Y_j(t). 
\end{align*}

The inclusion probabilities not varying with  time, it is straightforward to see, as in \cite{beaumont2013unified}, that  the term in parentheses at the right-hand side of (\ref{sampling_error}) is zero for Poisson sampling. Moreover, as shown in the Appendix, under broad assumptions upon the inclusion probabilities which are satisfied by the simple random sampling without replacement and fixed-size high-entropy designs, the term in parentheses at the right-hand side of (\ref{sampling_error}) is negligible in the sense that 
\[
\sup_{t \in [0,T]} \left| \sum_{i\in s}d_iA_i(t)-\sum_{i\in U}A_i(t) \right| =O_p(n^{-1/2}).
\]
Thus,  we can consider that
\begin{align}
\hat{t}_Y(t) &\simeq t_Y(t)+\sum_{i \in s} B^{HT}_{1i}(t) +  \sum_{i \in U-s } B^{HT}_{0i}(t), \quad t\in [0,T].   
\label{def:approxtB}
\end{align}
The first term at the right-hand side of  previous approximation is not random. Consequently,  the precision of the estimator will be influenced only by the two other terms in~(\ref{def:approxtB}). The conditional bias of a particular unit can thus be interpreted as the contribution of this unit to the sampling error.  An influential unit is defined as a unit with a large conditional bias and the idea is to downplay the impact of such units at the right-hand side of (\ref{sampling_error}). 
A new challenge, compared to the univariate framework studied in \cite{beaumont2013unified} comes  from the fact that  the conditional bias is now a function of time and as we can note in  Figure~\ref{sanscorrection}, the shape and the values of the conditional bias can be very different from one individual, or time instant, to another. Different ways of dealing with this issue are developed in the following.

\section{Point-wise robust estimators} \label{sec:pointwr}

A first possibility is to  directly apply the method of \cite{beaumont2013unified} at the $D$ instants $t_1, \ldots, T_D$. Considering the Huber function,   $\psi_{c}(z)=\mbox{sgn}(z) \mbox{min}(|z|,c)$ which depends on the tuning  constant $c>0$, with $\mbox{sgn}(z)=1$ if $z\geq 0$ and $-1$ otherwise, %.   Assuming that the term in parentheses on the right-hand of (\ref{sampling_error}) is negligible,
 we can construct the following point-wise robust estimator  of $t_Y(t)$:
\begin{eqnarray}
%\hat t_y^R(t)&= &t_y(t)+\sum_{i\in s}\psi_{c(t)}(B^{HT}_{1i}(t))+\sum_{i\in U-s}B^{HT}_{0i}(t)\nonumber\\
\hat t_Y^{(R1)}(t) &=& \hat t_Y(t)+\sum_{i\in s}\psi_{c(t)} \left(B^{HT}_{1i}(t) \right)-\sum_{i\in s}B^{HT}_{1i}(t)\label{robust_estim1}\\
& = & \hat{t}_Y(t)+\Delta(c(t)), \quad \mbox{for all} \quad t\in [0,T].\label{def:deltact}
\end{eqnarray}
So, for a given value $c(t)$ those conditional bias $B^{HT}_{1i}(t)$ larger than $c(t)$ will be cut-off at $c(t)$ in the second-term at the right-hand  side of  (\ref{robust_estim1}). Clearly, the efficiency of the robust estimator depends on the choice of the tuning constant $c(t)$. As $c(t)$ increases, the estimator becomes closer to the non robust estimator. The new estimator $\hat{t}_Y^{(R1)}(t)$ is biased but of smaller variance than that of the non robust one, so we hope to improve the global precision measured by the mean squared error. The trade-off between variance and bias is controlled again by the tuning constant $c(t)$: a large value for $c(t)$ implies small bias but large  variance and a small value for $c(t)$ implies large bias and small variance.

\subsection{Minimax approach for choosing the optimal tuning constant}
We determine the optimal tuning constant in a pointwise manner, namely we determine for each $t$, the optimal value $c_{opt}(t)$ is chosen according to the minimax approach suggested by \cite{beaumont2013unified}. The value $c_{opt}(t)$,  that is not necessarily unique, satisfies
\begin{align}
c_{opt}(t) & = \mbox{arg}\min_{c \geq0} \max_{i\in s} \left| \hat B^{RHT}_{1i}(c(t)) \right|,
\label{def:copt}
\end{align}
where $\hat B^{RHT}_{1i}(c(t))$ is the estimator of the conditional bias of the robust estimator $\hat t^{R}_{y}(t)$. Using relation (\ref{def:deltact}), the conditional bias of the robust estimator is $B^{RHT}_{1i}(c(t))=B^{HT}_{1i}(t)+E_p(\Delta(c(t))|I_i=1)$ and can be estimated by 
\begin{eqnarray*}
\hat{B}_{1i}^{RHT}(t)=\hat{B}_{1i}^{HT}(t)+\Delta(c(t)).
\end{eqnarray*}
Following \cite{beaumont2013unified}, the optimal value of $\Delta(c(t))$ is
\begin{align*}
\Delta \left(c_{opt}(t) \right) &=-\frac{1}{2}\left(\hat{B}^{HT}_{min}(t)+\hat{B}^{HT}_{max}(t) \right),
\end{align*}
where $\hat{B}^{HT}_{min}(t)=\mbox{min}_{i\in s}\hat B^{HT}_i(t)$ and $\hat{B}^{HT}_{max}(t)=\mbox{max}_{i\in s}\hat B^{HT}_i(t)$ are the minimum and  respectively, the maximum of the estimated absolute conditional biases $\hat{B}_i(t)$ over the sample. 
The optimal robust estimator is therefore, at each instant $t$,
\begin{align}
\hat{t}_{Yopt}^R(t) & = \hat{t}_Y(t)+\Delta(c_{opt}(t))\nonumber\\
 & =  \hat{t}_Y(t)-\frac{1}{2} \left(\hat{B}^{HT}_{min}(t)+\hat{B}^{HT}_{max}(t) \right).
 \label{def:hattoptR}
\end{align}

\begin{figure}[ht]
\begin{center}
\includegraphics[scale=0.75]{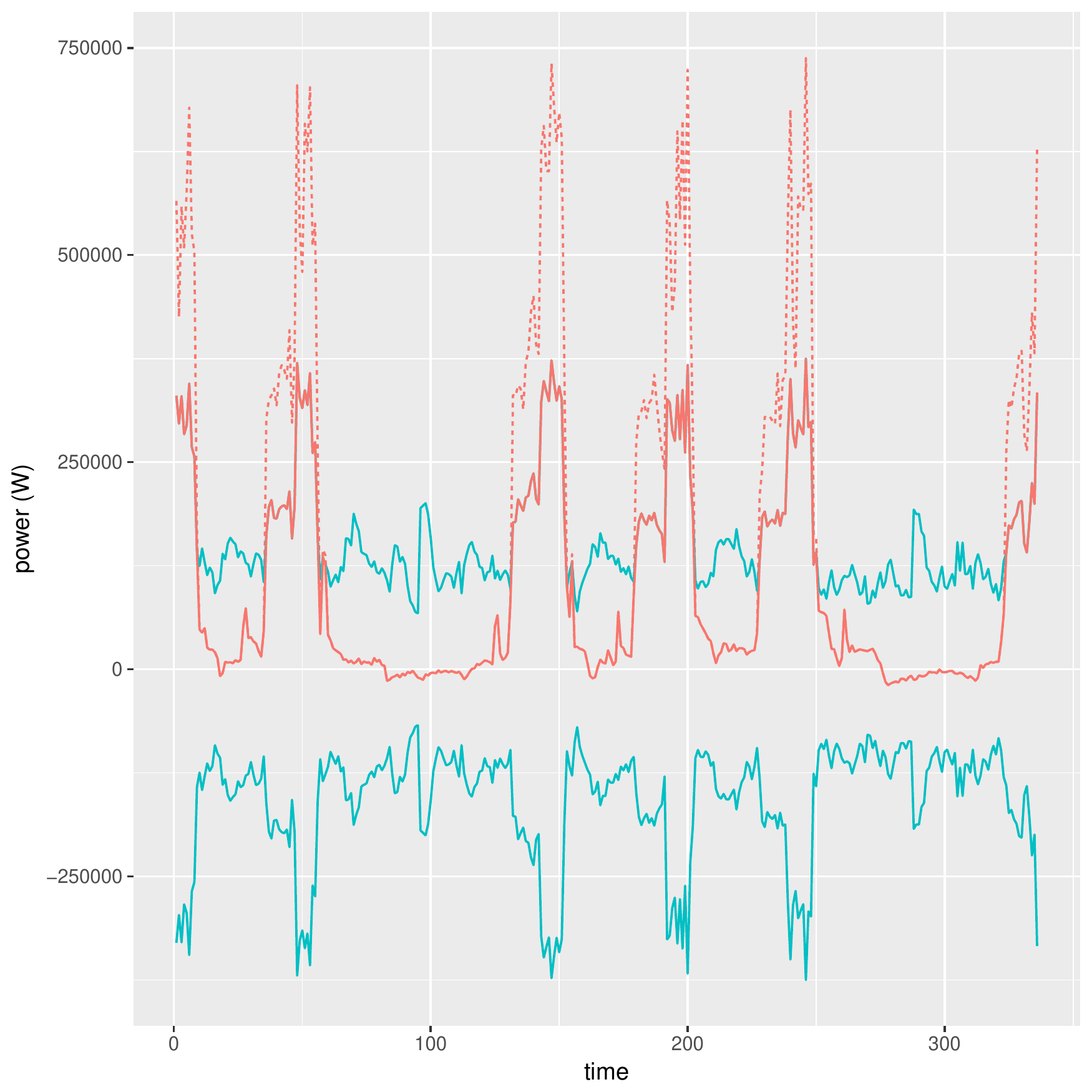} 
\caption{An estimated conditional bias (red solid line) and the same conditional bias truncated (red dotted line) for cut-off limits (blue lines) defined point-wisely.}
\label{apreslacorrectionponctuel}
\end{center}
\end{figure}

Remark that the optimal tuning constant $c_{opt}(t)$ varies over time, but there is not need to compute it in order to construct the optimal robust estimator. Note also that this method is essentially univariate since it deals independently with the different points in time and thus does not necessarily preserve the correlations between instants in the estimated total curves. We can think that  some information is lost by not making use of the strong temporal correlations between the values of $Y$ at different times. A robust estimator which takes into account  such possible correlations is presented in Section~\ref{rfpca}.

\subsection{A new criterion for choosing the optimal tuning constant}\label{qth_criterion}

We suggest minimizing the  sum over the sample of the $q$th power of the absolute value of $\widehat B^{RHT}_{1i}(c(t)),$ the conditional bias of the robust estimator.  This means that, for all $t$, we look for the optimal constant $c(t)$  satisfying the following criterion: 
\begin{eqnarray}
\label{pbNewInst}
c_{opt}^{alt}(t) &=&  \mbox{arg}\min_{c \geq 0}  \sum_{i \in s }  \left| \widehat B^{RHT}_{1i}(c(t))\right|^q, \nonumber\\
 &= & \mbox{arg}\min_{c \geq 0}  \sum_{i \in s }  \left| \widehat B^{HT}_{1i}(c(t)) +\Delta(c(t))\right|^q,  %\quad t \in [0,T],
\end{eqnarray}
where $\Delta(c(t))$ is given in (\ref{def:deltact}) and $q$ is a positive constant. The optimal solution may be found by numerical algorithms such as Newton-Raphson.  
 
By using this criterion, we penalize the conditional bias computed for the whole sample of individuals not only the maximum and the minimum of $ \widehat{B}^{RHT}_{1i}(c(t))$ as in \cite{beaumont2013unified}. In this way, each  $\left| \widehat{B}^{RHT}_{1i}(c(t))\right|$ for $i\in s$ will contribute to the optimisation research and as a consequence, the presence  in the sample at one instant $t$ of  a unit with very high influence will not cause a sudden change in the function $c(t)$ as it was the case with the minimax approach. 
  
Large values of $q$ will penalize large absolute values of $| \widehat B^{RHT}_{1i}(c(t))|$ while values of $q$ between $0$ and $1$ will penalize small conditional bias $| \widehat B^{RHT}_{1i}(c(t))|.$ So, for large $q,$ this new criterion will be close to the minimax criterion but with better regularity properties allowing the use of bootstrap methods in order to estimate the variance. To obtain robust estimates we thus advise to consider large values of $q$ ($q \geq 4$) that will ensure that high values of the conditional bias are sufficiently penalized.   
Note that for $q=1$, we obtain the median conditional bias curve which is not of interest here because it is robust to outlying (extremely large or small) values.   As a consequence, it will not be affected by these values and it can not be used to truncate outlying values.

\section{Robust estimation based on dimension reduction}

In a functional data setting, it is very common to use dimension reduction methods. In this paper, we use robust functional principal analysis and projection on basis functions such as wavelet function in order to transform the functional robust estimation issue into a series of univariate robust estimation issues.

\subsection{Spherical principal components analysis}\label{rfpca}

Principal components analysis is a popular tool to explore and to represent graphically the variations around their barycenter of multivariate and functional data (see \cite{Jolliffe2002} and \cite{Ramsay_Silverman_Livre} as well as \cite{cardot2010properties} for a presentation in a finite population setting). The aim is to  build new non-correlated variables, called principal components, that are linear combinations of the initial variables and of maximum variance. The principal components   are obtained via  the eigenfunctions of the covariance function $\gamma$ of the data $Y_i,  i=1, \ldots, N,$
\[
\boldsymbol{\gamma}(r,t)=\frac{1}{N}\sum_{i=1}^N(Y_i(r)-\overline{Y}_U(r))(Y_i(t)-\overline{Y}_U(t)), \quad r,t\in [0,T],
\]
where $\overline{Y}_U=N^{-1}\sum_{i\in U}Y_i$ is the mean, or the center, of the data.  However, it is well known that  the mean and the covariance are highly sensitive to outlying units  and consequently, principal components are also known to be highly non robust. 

 We consider now a robust version of PCA named spherical PCA (see \cite{locantore1999robust}) that has nice properties (see \cite{Gervini2008}) and is easy to compute. It consists in considering the eigenfunctions of the following sphericised ``covariance'' function
 \begin{eqnarray}
 \boldsymbol{\Gamma}(r,t)=\frac{1}{N}\sum_{i=1}^N\frac{Y_i(r)-m_N(r)}{||Y_i-m_N||}\cdot\frac{Y_i(t)-m_N(t)}{||Y_i-m_N||}, \quad r,t \in [0, T]\label{spherical_covariance}
 \end{eqnarray}   
where  $m_N$ is a robust indicator of location and $\| . \|$ denotes the $L^2[0,T]$-norm ($\|Y\|^2 = \int_0^T Y^2(t) dt$). Considering the unit norm functions $(Y_i(t)-m_N(t)) /\|Y_i-m_N \|$ instead of  $Y_i(t)-\overline{Y}_U(t),$ we perform a kind of winsorisation of the outlying  curves $Y_i$.
 As in \cite{locantore1999robust}, we use the geometric median  (see \cite{Kemperman1987} or \cite{Small1990}) as a robust location parameter of a set of points belonging to the space $L^2[0,T]$.
With a finite population point of view, the median curve of  the elements $Y_1, \ldots, Y_N$,  is defined by: 
\begin{align}
m_N &=\mbox{arg}\min_{ y \in L^2[0, T]}\sum_{i=1}^N\| Y_i-y\| .
\label{def_med}
\end{align}
The relation (\ref{def_med}) arises as a natural  generalization of the well-known characterization of the univariate median. It is also called the \textit{spatial median} (\cite{Brown}) because, from a geometric point of view, the median is the point that minimizes the sum of distances to the  points in the population.  The names \textit{$L_1$-median} (\cite{Small1990})  and \textit{geometric median} (\cite{chaudhuri_1996}) have been also employed for  $m_N$.

If we assume  that  $Y_i$, for  $i=1,\ldots, N,$ are not concentrated on a line,  the median exists and is unique (see \cite{Kemperman1987}). If $m_N\neq Y_i$ for all $i=1,\ldots, N,$ then it is the unique  solution of the following estimating equation:
 \begin{align}
 \sum_{i=1}^N\frac{Y_i(t) - m_N(t)}{||Y_i - m_N||} &=0, \quad t\in [0, T]
 \label{medU}
 \end{align}
and it may be computed by using fast iterative algorithms such as  Weiszfeld's algorithm  (see \cite{Weiszfeld37} and \cite{VZ00}) for multivariate data or gradient algorithms (see \cite{Gervini2008}) for sparse functional data.

Then, performing spherical PCA consists in computing the eigenvalues $\lambda_j$ and the corresponding orthonormal eigenfunctions $\mathbf{v}_j$ of the covariance $ \boldsymbol{\Gamma}$ of these projected data instead of the initial data. As for the location estimate, the influence of the outlying observations can be greatly reduced. Furthermore \cite{Gervini2008} shows that if the distribution of $Y_i$ is symmetric, then the covariance $\boldsymbol{\gamma}$ and the spherical covariance $\boldsymbol{\Gamma}$ have the same orthonormal eigenfunctions $\mathbf{v}_j,$ $j=1, \ldots, N$.

The curves $Y_i$ in the population can also be approximated, in this new orthonormal basis,  leading to a kind of robust  Karhunen-Loeve expansion, that allows to get the best approximation of  $Y_i(t)-m_N(t)$ in a finite $K$-dimensional space (see \cite{Ramsay_Silverman_Livre}):
\begin{align}
Y_i(t) &= m_N(t) + \sum_{k=1}^K \langle Y_i-m_N,\mathbf{v}_k \rangle \mathbf{v}_k(t)+\epsilon_i(t), \quad \mbox{for}\quad i\in U,
\end{align}
where $\epsilon_i(t) = Y_i(t) - m_N(t) - \sum_{k=1}^K \langle Y_i-m_N,\mathbf{v}_k\rangle \mathbf{v}_k(t)$ is a remainder term and $\langle \cdot,\cdot \rangle$ is the inner product in $L^2[0,T]$. Here, $ \langle Y_i-m_N,\mathbf{v}_k \rangle \mathbf{v}_k(t)$ is the projection of the centered curve $Y_i-m_N$ onto the rank one space generated by  function $\mathbf{v}_k$. For our purpose, we consider the same (large enough) value of $K$ for all the curves $Y_i$  to keep most of the variation in the data.

With these considerations, the approximation of the total curve  in a finite $K$-dimensional space is given by
\begin{align}
t_Y(t) & \approx  N m_N(t) +   \sum_{k=1}^KF_k\mathbf{v}_k(t),
\label{KL-total}
\end{align}
where
\begin{align*}
F_k &=\sum_{i \in U} \langle Y_i-m_N,\mathbf{v}_k \rangle, \quad \mbox{for}\quad k=1, \ldots, K,
\end{align*}
is the population total of the projections on $\mathbf{v}_k$ of the "centered" data $Y_i-m_N$.  So, we can write the finite population total $t_Y$ as the sum of a robust location parameter, the median $m_N$, and the sum of $K$ products between the robust eigenfunctions $\mathbf{v}_k(t)$ and the real coordinates $F_k$ in this new basis. 
The interest of considering decomposition (\ref{KL-total}) is that the total of a function with a continuous time index  is decomposed into a new multivariate problem in which robustification techniques can be applied to each real component. 

\subsubsection{Estimation of the robust principal components} 
In order to estimate $t_Y,$ we need to estimate first the median and the eigenfunctions $\mathbf{v}_k$ for all  $k=1, \ldots, K$. 
A natural estimator of the geometric median $m_N$ is given by the solution $\hat m$ of the following non linear estimating equation (see   \cite{chaouch_goga_2012}),
 \begin{align}
 \sum_{i\in s}d_i\frac{Y_i(t) - \hat m(t)}{||Y_i - \hat m||} &=0, \quad t\in [0, T].
 \label{estim_median_sondage}
 \end{align} 
Numerically, the solution is generally reached in a few iteration of a weighted version Weiszfeld's algorithm.

The spherical covariance function given in (\ref{spherical_covariance}) is estimated as follows
 \begin{eqnarray}
 \widehat{\boldsymbol{\Gamma}}(r,t)=\frac{1}{N}\sum_{i\in s}d_i\frac{(Y_i(r)-\hat {m}(r))}{||Y_i-\hat{m}||}\cdot\frac{(Y_i(t)-\hat {m}(t))}{||Y_i-\hat{m}||}, \quad \mbox{for all}\quad r,t \in [0, T],
 \end{eqnarray}
where $\hat m$ is the estimator of the median $m_N$ given in (\ref{estim_median_sondage}). Then, estimators of the eigenvalues $\lambda_j$ of $\boldsymbol{\Gamma}$ with the associated eigenfunctions $\hat{\mathbf{v}}_j, j=1\ldots, N$ are obtained  by the spectral decomposition of the estimated covariance $\widehat{\boldsymbol{\Gamma}}(r,t)$.

A natural estimator of the approximation of $t_Y$ given in (\ref{KL-total}) is obtained by replacing the unknown quantities with their estimators: 
\begin{align}
\hat t_Y^{(2)}(t)&=N\hat m(t)+ \sum_{k=1}^K\hat F_k\hat{\mathbf{v}}_k(t),
\label{non_robust1}
\end{align}
where $\hat F_k=\sum_{i \in s}d_i \langle Y_i-\hat m,\hat{\mathbf{v}}_k \rangle$
is the substitution estimator for $F_k$. Note that even if  $\hat m(t)$ and $\hat{\mathbf{v}}_k(t)$ are robust estimates, the estimator given in (\ref{non_robust1}) is not robust because the coordinates $\hat F_k$, $k =1, \ldots, K$  are not robust.

\subsubsection{Robustifying the coordinates in the spherical PCA basis} \label{sec:estsphpca} 

We suggest to build the following robust estimates of the coordinates
\begin{align}
\hat{F}_k^R &=\hat{F}_k -  \sum_{i \in s}\hat{B}^{F}_{1i,k} + \sum_{i \in s} {\psi_{c_k}}(\hat{B}^{F}_{1i,k}), \quad k=1, \ldots, K,
\end{align}
where $\hat{B}^{F}_{1i,k}=\displaystyle\sum_{j\in s}\left(\frac{\pi_{ij}-\pi_i\pi_j}{\pi_j\pi_{ij}}\right) \langle Y_j-\hat{m},\hat{\mathbf{v}}_k \rangle$ is the estimator of the conditional bias of  $\hat F_k$, and $\psi_{c_k}$ is  the Huber function depending on the tuning constant $c_k$. An optimal value for $c_k$ may be found by using the minimax criterion or the new criterion defined in (\ref{pbNewInst}). Finally,  the second robust estimator of $t_Y$ is defined as follows
\begin{align}
\hat t_Y^{(R2)}(t) &=N\hat m(t)+ \sum_{k=1}^K\hat F^{R}_k\hat{\mathbf{v}}_k(t), \quad t \in [0,T].
 \label{robust_2}
\end{align}

\subsection{Projection on wavelet basis}
\label{remOndelettes}
Instead of using principal components, we may project data onto a basis of functions  $\phi_1, \dots, \phi_Q$ which do not depend on the data. Electricity load curves are known to be irregular, as seen in~Figure~\ref{atrandom}, and natural candidates are wavelet basis (see \cite{Mallat98}). 

The curves $Y_i, i\in U$  may be expanded  as follows
\begin{align*}
Y_i(t) &= \sum_{q =1}^{Q} a_{iq} \phi_q (t) + \epsilon_i(t), \quad t\in [0,T],
\end{align*}
where $\epsilon$ is an approximation residual. Note that, unlike the principal component analysis, the functions $\phi_1, \dots, \phi_Q$ are known and do not need to be estimated. The coefficients $a_{iq},$ for $q=1, \ldots, Q,$ depend on $Y_i$ and  are unknown for the non-sampled individuals. 
As in  robust principal component analysis, the total curve may be approximated by 
\begin{align}
t_Y(t) \simeq& \sum_{q=1}^Q t_{a_q}\phi_q(t), \quad t\in [0,T],
%t_{Y}(t) \simeq &  \sum_{q=1}^Q  t_{\alpha_q} \phi_q(t).
\end{align}
where $t_{a_q}=\sum_{i \in U}  a_{iq}$ is the unknown real population total of the coefficients $a_{iq}, $ for $q=1, \ldots, Q.$ The Horvitz-Thompson estimator of this new approximation of the total $t_Y$ is given by 
\[
\hat t^{(3)}_{Y}(t)=\sum_{q=1}^Q\hat{t}_{a_q}^{HT} \phi_q(t), \quad t\in [0,T],
\]
where $\hat{t}_{a_q}^{HT}= \sum_{i \in s} d_i a_{iq}. $ Robust estimators of $\hat{t}_{a_q}^{HT}$ may be built as above. Our third robust estimator of $t_Y$ is defined as follows:
\begin{align}
\hat{t}_{Y}^{(R3)} (t) &=   \sum_{q=1}^Q \hat{t}_{a_q}^{RHT} \phi_q(t), \quad t\in [0,T],
\label{robustR3}
\end{align}
where $\hat{t}_{a_q}^{RHT}=\hat{t}_{a_q}^{HT}+\sum_{i\in s}(\psi_{c_q}(\hat{B}_{1i,q})-\hat B_{1i,q})$ is the robust estimator of $\hat{t}_{a_q}^{HT},$ $q=1, \ldots, Q$, with $\hat B_{1i,q}$ the conditional bias of $\hat{t}_{a_q}^{HT}$ and $\psi_{c_q}$ the Huber function depending on the tuning constant $c_q$ whose value may be determined for each $\hat{t}_{a_q}^{HT}$.

\section{Global functional truncation methods based on statistical depth }
\label{fonctronc}

The aim of this section is to introduce a global way of truncating the conditional-bias curve. In order to do that, we use the notion of statistical depth which allows to  define an order relation in a set of curves, from the most central curve to the most outlying one. In the context of functional data, the depth may be defined in many different ways: see for example \cite{cuesta2006random}, \cite{gervini2012outlier}, \cite{fraiman2001trimmed} or \cite{hyndman2010rainbow}. Many of these notions of depth  are rather difficult to put into practice and are not considered here.
 In the following, we consider the modified band depth as defined by \cite{lopez2009concept} as well as  a depth notion based on the $L^2[0,T]$ distance from the center of the projected data onto the axis obtained by spherical PCA.

\subsection{Definition of the modified band depth (MBD)} \label{sec:depthmbd}
The Modified Band Depth (MBD), studied by \cite{lopez2009concept}, of a discretized curve is the number of times (or the proportion of time for continuous time observations) the curve, within a set of curves,  is ``lying between a couple of other curves'':
\begin{align*}
MBD_i &=\frac{1}{\binom{2}{n}}  \sum_{j,k \in s \; j \ne k} \frac{1}{D} \sum_{d=1}^{D} \1_{\left[\min(\hat B^{HT}_{1j}(t_d), \hat B^{HT}_{1k}(t_d)) \leq  \hat B^{HT}_{1i}(t_d) \leq \max(\hat B^{HT}_{1j}(t_d), \hat B^{HT}_{1k}(t_d))\right]} \\
& \approx \frac{1}{\binom{2}{n}}  \sum_{j,k \in s \; j \ne k}  \frac{1}{T} \int_0^T \1_{\left[\min(\hat B^{HT}_{1j}(t), \hat B^{HT}_{1k}(t)) \leq  \hat B^{HT}_{1i}(t) \leq \max(\hat B^{HT}_{1j}(t), \hat B^{HT}_{1k}(t))\right]} \ dt.
\end{align*}
This indicator takes into account the length of the time interval during which the curve $\hat B^{HT}_{1i}$ is not lying between each couple of other curves : a curve which is not included between others during a small time interval will be considered as ``less outlying'' than another one which is out during a longer period. So, the more often a curve is included  entirely between others  the more it is considered as central and by consequence, a curve with a high MBD will be considered as central.

\subsection{Central area based on MBD and robust estimator}
We compute the depth value $MBD_i$ of the conditional bias curve $\hat{B}^{HT}_{1i}(t)$ for all units $i$ belonging to the sample and let  $I$ be the  central region containing the 50\% of the deepest curves $\hat{B}^{HT}_{1i}(t), i\in s$. Let $L$ be the lower functional bound and  $U$ the upper functional bound computed over $I$, for $t \in [0,T]$:
\begin{align*}
L(t) =\min_{i \in I} \hat{B}^{HT}_{1i}(t) \quad \mbox{and}\quad
U(t) =\max_{i \in I} \hat{B}^{HT}_{1i}(t).
\end{align*}

The idea of using a 50\% central region  has been suggested first in the functional bagplot introduced by  \cite{hyndman2010rainbow} and in the functional boxplot by  \cite{sun2011functional}.

The conditional-bias curves entirely located inside these boundaries will not be modified whereas the curves taking  values outside the central region, for some period of time, will be truncated by using a truncation function $\psi$ as in the non-functional  case. An obvious candidate is the Huber function $\psi_c(y)=\mbox{max}(\mbox{min}(y, c), -c)$ depending on a tuning constant $c>0$ which can be easily generalized to take into account a region delimited by a lower and an upper delimiting curves:
\begin{align*}
\psi (\hat B^{HT}_{1i}(t)) &= \max\left(\min(\hat B^{HT}_{1i}(t), U(t)), L(t)\right) \quad \mbox{for all}\quad t\in [0,T].
\end{align*}
Remark that $L$ needs not to be $-U$. We propose to use the following truncation function,
\begin{align*}
\psi_{\alpha} (\hat B^{HT}_{1i}(t)) &=  \max \left( \min \left(\hat B^{HT}_{1i}(t), \mu_{\hat B}(t) + \alpha (U(t) - \mu_{\hat B}(t)) \right) , \mu_{\hat B}(t) + \alpha (L(t) - \mu_{\hat B}(t)) \right),
\end{align*}
where $\alpha$ is an unknown positive dilatation parameter that controls the size of the central region (\cite{sun2011functional}) and $\mu_{\hat B}$ is the  mean of the estimated conditional bias over the sample.
In practice, the delimiting curves $\mu_{\hat B}(t) + \alpha (U(t) - \mu_{\hat B}(t))$ and  $\mu_{\hat B}(t) + \alpha (L(t) - \mu_{\hat B}(t))$ are smoothed, using a mobile averaging technique, in order to  avoid a too irregular truncation. 

Figure~\ref{apreslacorrectionfonctionnel} displays the mechanism of  global truncation  based on modified band depth. 
The upper (U) and the lower (L) curves delimiting the $\alpha$ central area are plotted in blue. A conditional bias curve is plotted in red. Parts of this curve lying outside of the central area, plotted in red dotted line, will be truncated and replaced by the corresponding parts of the bound curves. We can remark that the central zone constructed in this way reflects the daily seasonality of the data.  In Figure \ref{comparaisondesbornes}, we plot central areas constructed according to the suggested methods: pointwise, spherical PCA and based on modified band depth. We can remark on this plot that the central area based on modified band depth is not symmetric and is much larger than the other two areas.

\begin{figure}[!h]
\begin{center}
\includegraphics[scale=0.75]{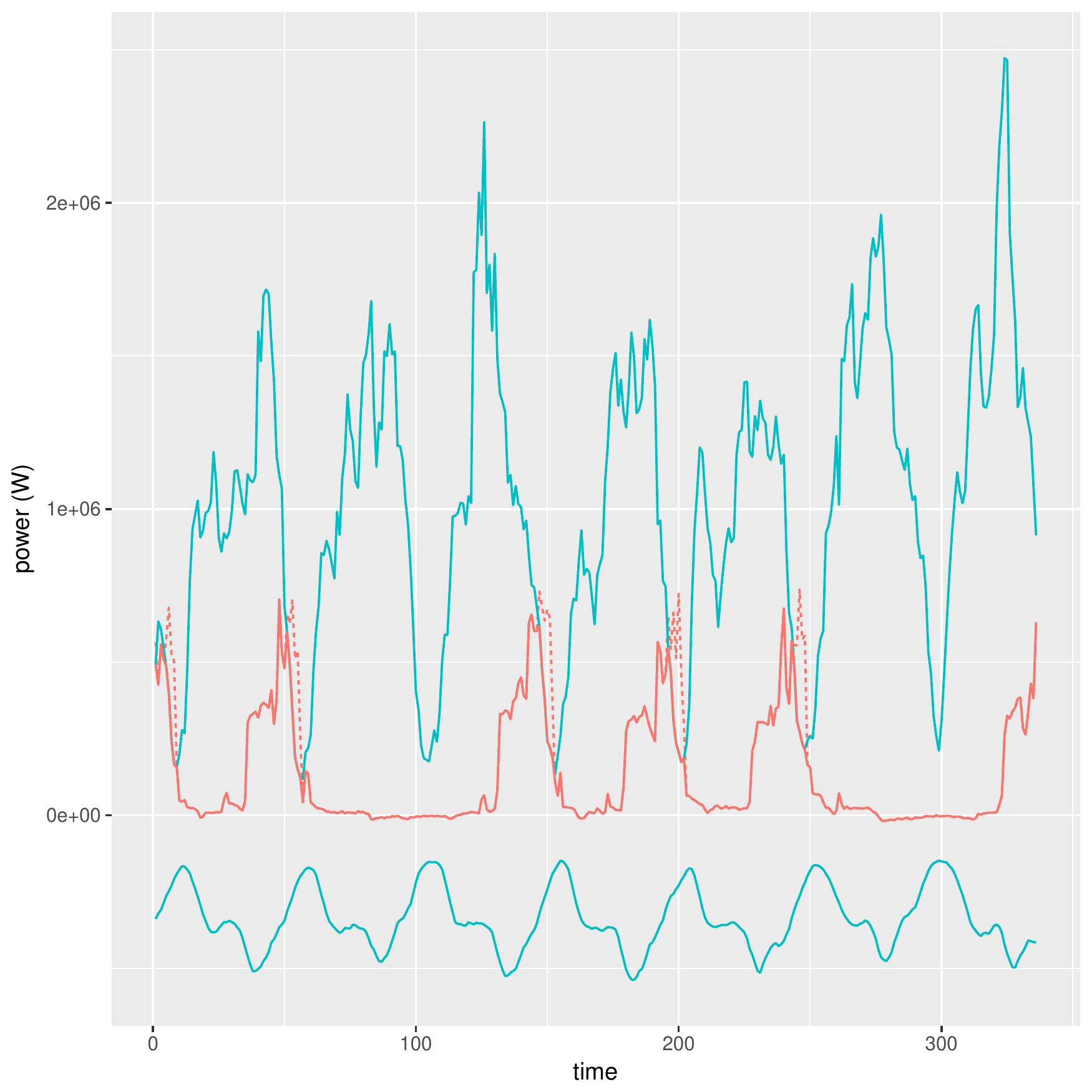} 
\end{center}
\caption{The upper (U) and the lower (L) curves delimiting the $\alpha$ central area are plotted in blue. A conditional bias curve is plotted in red and the truncated part of this curve is plotted in red dotted line. }
\label{apreslacorrectionfonctionnel}
\end{figure}

\begin{figure}[!h]
\begin{center}
\includegraphics[scale=0.8]{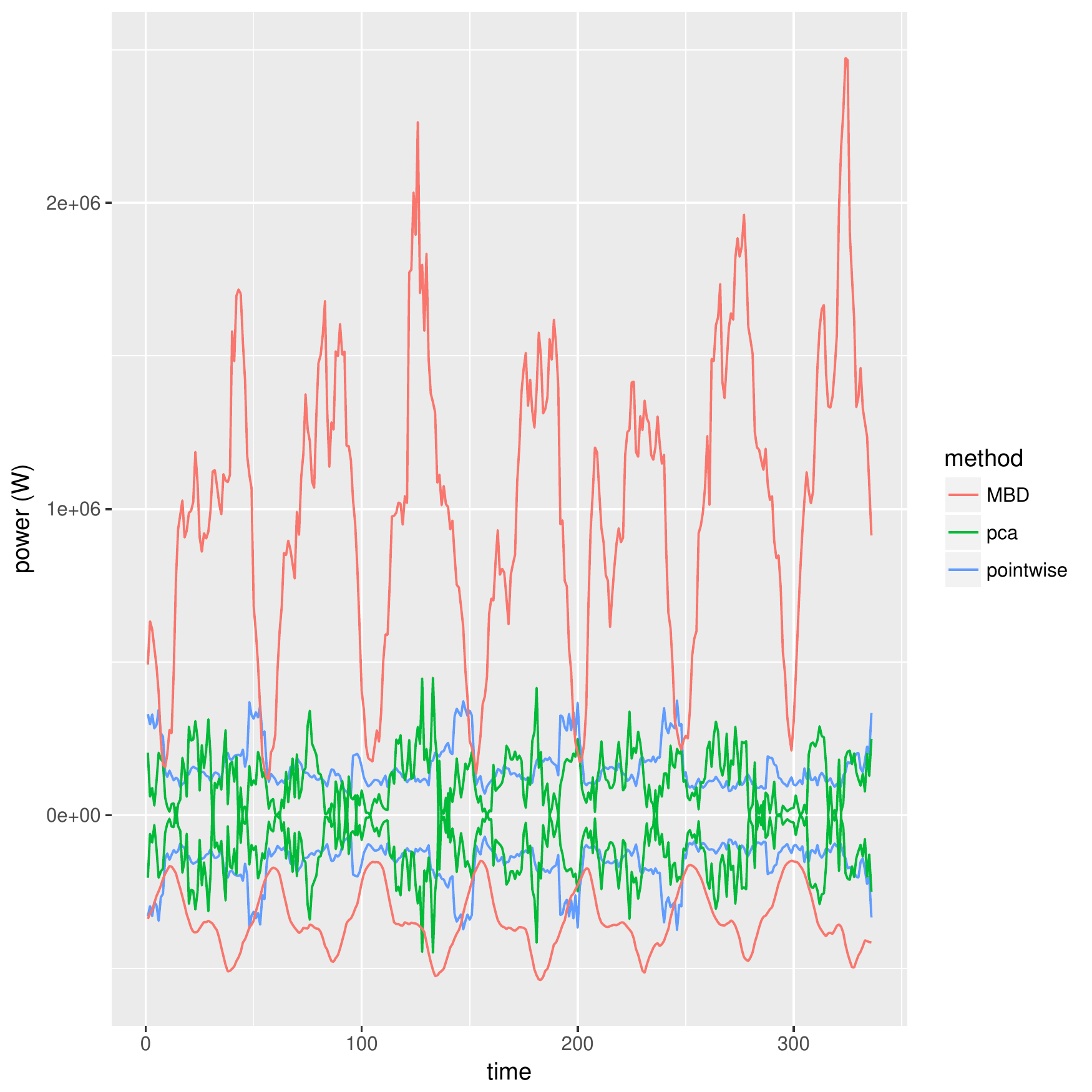} 
\end{center}
\caption{Central areas constructed according to three methods: point-wisely (blue line), robust PCA (red line) and functional modified band depth (green line).  }
\label{comparaisondesbornes}
\end{figure}

Using the global truncation function $\psi_{\alpha}$ and relation (\ref{robust_estim1}), we can construct a new robust estimator as follows  
\begin{align}
\hat t_{Y}^{(R4)} &=  \hat t_Y(t)+\sum_{i\in s}\psi_{\alpha} \left(B^{HT}_{1i}(t) \right)-\sum_{i\in s}B^{HT}_{1i}(t), \quad t\in [0, T].\label{MBD_estim}
\end{align}

The  dilatation factor $\alpha$ allows to control  the trade-off between bias and variance: for $\alpha$ small, the curves are strongly truncated meaning large bias and small variance whereas for $\alpha$ large, the curves are less truncated meaning less bias and larger variance. To determine the value of the truncation parameter $\alpha,$ we can use the functional minimax approach:
\begin{align*}
\alpha_{opt} &=  \mbox{arg}\min_{\alpha}\max_{i \in s} \frac{1}{D}\sum_{d=1}^D \left|\hat B^{HT}_{1i}(t_d) + \Delta_{\alpha}(t_d)\right|
\end{align*}
or the $q$th power criterion introduced in Section \ref{qth_criterion}:
 \begin{equation}
 \label{pbNewFonc}
\alpha_{opt}^{alt}=\mbox{arg}\min_{\alpha \geq 0} \frac{1}{D}\sum_{d=1}^{D} \sum_{i \in s} \left|\hat{B}_{1i}^{HT}(t_d) + \Delta_{\alpha}(t_d) \right|^q,
\end{equation}
where
\begin{align*}
\Delta_{\alpha}(t) &= \sum_{i \in s } \left(\psi_{\alpha}(\hat B^{HT}_{1i}(t)) - \hat B^{HT}_{1i}(t)\right). 
\end{align*}
The optimum values $\alpha_{opt}$ and $\alpha_{opt}^{alt}$ are obtained numerically by a Newton-Raphson algorithm.

\section{Mean square error estimation}

In this section, we derive approximate point-wise estimators of the mean square errors of the  robust estimators.
For a given time instant $t$, the mean square error (MSE) can be expressed as 
\begin{align*}
MSE_p\left(\hat{t}_Y^R(t)\right) &=V_p \left(\hat{t}_Y^R (t)\right) + E_p\left(\hat{t}_Y^R (t) - \hat t_Y (t) \right)^2-V_p(\hat{t}_Y^R (t) - \hat t_Y (t) ).
\end{align*}
Similarly to \cite{gwet1992outlier} and \cite{beaumont2013unified}, we suggest the following point-wise mean square error estimator:
\begin{align}
\widehat{MSE}_p(t)=v_p \left(\hat{t}_Y^R(t) \right) + \max \left[0, \left(\hat{t}_Y^R(t) - t_Y(t) \right)^2 - v_p\left(\hat{t}_Y^R(t) - \hat{t}_Y(t) \right) \right].
\label{GWformula}
\end{align}
where $v_p \left(\hat{t}_Y^R(t) \right)$ and $v_p\left(\hat{t}_Y^R(t) - \hat{t}_Y(t)\right)$ are design-consistent estimators of  $V_p \left(\hat{t}_Y^R (t)\right)$ and $V_p(\hat{t}_Y^R (t) - \hat t_Y (t))$.

Using relation (\ref{robust_estim1}), we can write the robust estimator $\hat{t}_Y^R(t)$ and $\hat{t}_Y^R(t) - \hat{t}_Y(t)$ as follows
\begin{eqnarray*}
\hat{t}_Y^R(t) &= &\sum_{i\in s}d_i(Y_i(t)+Z_{ic(t)}(t)) \quad \mbox{and}\\
\hat{t}_Y^R(t) - \hat{t}_Y(t) & = & \sum_{i\in s}d_iZ_{ic(t)}(t),
\end{eqnarray*}
where $Z_{ic(t)}(t)=\pi_i(\psi_{c(t)} (B^{HT}_{1i}(t)) -B^{HT}_{1i}(t)). $ 
For simple sampling designs for which  the first and second order inclusion probabilities are known, we can use the  Horvitz-Thompson variance estimator,
\begin{align}
v_p\left(\hat{t}_Y^R(t)\right) &=\sum_{i \in s}\sum_{j\in s} \frac{\pi_{ij}-\pi_i\pi_j}{\pi_{ij}} \frac{Y_i(t)+\hat {Z}_{i,c(t)}(t)}{\pi_i}\frac{Y_j(t)+\hat{Z}_{j,c(t)}(t)}{\pi_j},\label{var_estim_HT}
\end{align}
where $\hat {Z}_{ic(t)}(t)=\pi_i(\psi_{c(t)} (\hat{B}^{HT}_{1i}(t)) -\hat{B}^{HT}_{1i}(t)). $ A variance estimator is obtained for  $\hat{t}_Y^R(t) - \hat{t}_Y(t)$ by a similar procedure. For the robust estimator $\hat t_Y^{(R4)}$ given in (\ref{MBD_estim}) (section \ref{sec:depthmbd}) computed by using functional truncation methods based on depth, a variance estimator may be computed by using (\ref{var_estim_HT}) with $\hat{Z}_{i\alpha} (t) = \pi_i \left[ \psi_{\alpha}(\hat{B}^{HT}_{1i}(t))- \hat{B}^{HT}_{1i}(t) \right]$.

Using linearization techniques, we can write for the robust estimator $\hat t_Y^{(R2)}$ given in (\ref{robust_2}):
\begin{eqnarray*}
N^{-1} \left(\hat t_Y^{(R2)}(t)-t_Y \right) %\hat m-m+N^{-1}\sum_{k=1}^K\left(\hat F^R_k\hat{\mathbf{v}_k}-F^R_k\mathbf{v}_k\right)\\
 & \simeq & N^{-1}\left(\sum_{i\in s} d_iu_i-\sum_{i\in U}u_i\right)+N^{-1}\sum_{k=1}^K\left[(\hat F^R_k- F_k)\mathbf{v}_k+(\hat{\mathbf{v}}_k-\mathbf{v}_k)F_k\right]
\end{eqnarray*}
where $u_i=N\boldsymbol{\Gamma}^{-1}\left(\frac{Y_i-m}{||Y_i-m||}\right), i\in U$ is the linearized variable of $m$ (see \cite{chaouch_goga_2012}) with $\boldsymbol{\Gamma}$ given in (\ref{spherical_covariance}). We also have
 \begin{eqnarray*}
 N^{-1}\sum_{k=1}^K\left(\hat{\mathbf{v}}_k-\mathbf{v}_k\right)F_k &\simeq& N^{-1}\left(\sum_{i\in s}d_i\sum_{k=1}^KF_k\tilde{\mathbf{v}}_{i,k}-\sum_{i\in U}\sum_{k=1}^KF_k\tilde{\mathbf{v}}_{i,k}\right),
 \end{eqnarray*}
 where $\tilde{\mathbf{v}}_{i,k}=\sum_{\ell\neq i}<Y_k-m,\mathbf{v}_i><Y_k-m,\mathbf{v}_\ell>\mathbf{v}_\ell/(\lambda_i-\lambda_{\ell})||Y_k-m||^2$ is the linearized variable of $\mathbf{v}_{k}$ obtained with similar arguments as in \cite{cardot2010properties}. We also have
  \begin{eqnarray*}
& & N^{-1}\sum_{k=1}^K(\hat F^R_k- F_k)\mathbf{v}_k\\
&\simeq& N^{-1}\sum_{i\in s}d_i\sum_{k=1}^K\left(<Y_i-m,\mathbf{v}_k>+\pi_i(\psi_{c_k}(B^F_{1i,k})-B^F_{1i,k})\right)\mathbf{v}_k-N^{-1}\sum_{i\in U}\sum_{k=1}^K<Y_i-m,\mathbf{v}_k>\mathbf{v}_k.
 \end{eqnarray*}
The variance estimator can then be computed for 
 $$
 \hat{Z}_{ic}(t)=\hat u_i(t)+\sum_{k=1}^K\left(\hat{F}_k\hat{\tilde{\mathbf{v}}}_{i,k}+<Y_i-\hat m,\hat{\mathbf{v}}_k>\hat{\mathbf{v}}_k+\pi_i(\psi_{c_k}(\hat{B}^F_{1i,k})-\hat{B}^F_{1i,k})\hat{\mathbf{v}}_k\right).
 $$
For the third robust estimator $\hat t^{(R3)}_Y$ given in (\ref{robustR3}) based on projection on known basis function $\phi_1, \ldots, \phi_Q$, the variance estimator is obtained for $Z_{ic}(t)=\sum_{q=1}^Q(Y_i+\pi_i(\psi_{c_q}(\hat{B}_{1i,q})-\hat B_{1i,q}))\phi_q(t). $ 

\subsection*{Bootstrap}
Approximation by bootstrap of the variance estimator $v_p$ used  in (\ref{GWformula}) is possible. We consider the without replacement bootstrap  introduced by \cite{gross_1980} for simple random sampling without replacement and that can be extended easily to stratified simple random sampling. The method consists in creating a pseudo-population $U^*$ by duplicating each unit $i\in s,$ $d_i=1/\pi_i$ times. Several methods have been proposed to deal with the situation when $d_i$ is not integer. We consider here the population bootstrap as suggested by \cite{booth1994bootstrap} which consists in completing $U^*$ by a simple random sampling of size $N-[N/n]$. From this pseudo-population, we select $B$ replication samples $s^*$ of size $n$ according to the initial sampling design. The bootstrap variance estimator of the robust estimator $\hat{t}_Y^R(t)$ is the empirical variance of $\hat{t}_Y^{R}(t)$ computed over the replication samples:
\begin{eqnarray}
v^{boot}_p(\hat{t}_Y^{R}(t))=\frac{1}{B-1}\sum_{b=1}^B\left(\hat{t}_Y^{R,b}(t)-\frac{1}{B}\sum_{b=1}^B\hat{t}_Y^{R,b}(t)\right)^2.\label{var_boot}
\end{eqnarray}

The value of the cut-off tuning parameter  $c$  is computed in each replication using the minimax approach. However, as the robust estimator based on the minimax approach is built using minima and maxima (which are "non linearizable" functions), we may have poor results for estimates based on  population bootstrap.

We also consider  the generalized bootstrap studied by \cite{bertail1997bootstrap}. For this bootstrap method, the sample of individuals is kept unchanged  but  the sampling weights are replicated. More precisely, we generate random weights $w^{*b}_i $ with $b=1, \ldots, B$ and $B$ large, such as $E(w^{*b}_i )=N^{-1}$, $\mbox{Var}(w^{*b}_i)=(1-\pi_i)N^{-2}$ and $\mbox{Cov}(w^{*b}_i,w^{*b}_j)=(1-\pi_i\pi_j/\pi_{ij})N^{-2}, i\neq j$. In practice, $w^{*b}_i$ may be simulated from a multivariate normal law with moments given above. The parameters of interest are written as functions of means and  means of type $\mu_Y=\sum_{i\in U}Y_i/N$, that estimated at each replication $b$ by $\hat {\mu}^{b}_Y=\sum_{i\in s}w^{*b}_id_iY_i. $  Formula (\ref{var_boot}) is next used to obtain a variance estimator of the robust estimators suggested in this paper.

\section{An illustration with real dataset}\label{section_application}

The methods and estimators studied in this paper are illustrated on data from the Irish Commission for Energy Regulation (CER) Smart Metering Project that was conducted in 2009-2010 (CER, 2011)\footnote{The data are available on request at the address: \\ \texttt{http://www.ucd.ie/issda/data/commissionforenergyregulation/}}.  This dataset contains thousands of electricity load curves of residential clients observed every half-hour during one year. We have selected from this dataset  $N=3994$ load curves without missing data and the electricity consumption recorded over one week, from the 18th to the 24th of January 2010. So, we have $D=336$ points in time. The interest parameter is the total consumption electricity during this week.

We consider two sampling designs: simple random sampling (SRS) without replacement and stratified random sampling with SRS within strata (STR). For the stratified sampling, strata are built by considering the total electricity consumption over the second semester of 2009. We have built 5 strata, containing respectively 1270, 898, 770, 659 and 397 statistical units. The first strata corresponds to meters with small levels of consumption whereas the last one is associated to the meters with the largest levels of consumption.  In this scenario, there are no "strata jumpers".  We consider two other STR samplings with $10\%$ strata jumpers (STR-SJ10) and respectively, with $20\%$ "strata jumpers" (STR-SJ20). These "strata jumpers" are simulated by selecting randomly, with equal probabilities, some units in the population and then affecting them to a wrong stratum, which is also chosen randomly with equal probabilities.
For each  scenario, we have considered three sample sizes: $n=40, 100$ and respectively, $n=400$ and the sample sizes within strata are computed according to the optimal allocation taking the consumption of the previous week as auxiliary information.

\subsection{Performance of the suggested robust estimators of the total consumption curve}
We evaluate the performances  and compare the different estimators presented in previous sections for various situations: different sampling designs, "strata jumpers" rates, sample sizes. The estimators considered  here are:
\begin{itemize}
\item the Horvitz-Thompson (HT) estimator;
\item the point-wise robust estimator $\hat t^{(R1)}_Y$ given by (\ref{robust_estim1}) with the tuning constant $c$ chosen by the minimax pointwise criterion (minimax pointwise) and the $q$th ($q=4, 10)$ power criteria (qth pointwise); 
\item the robust estimator $\hat t^{(R2)}_Y$ given by (\ref{robust_2}) and based on spherical PCA with the minimax criteria (robust PCA) and $K=5$ principal components;
\item the robust estimator $\hat t^{(R3)}_Y$ given by (\ref{robustR3}) and based on wavelet expansions with the minimax criteria (robust wave)\footnote{wavelets Daubechies Least Asymetric, 10};
\item the robust estimator $\hat t^{(R4)}_Y$ given by (\ref{MBD_estim}) with the global truncation function based on the modified band depth,  minimax criteria (MBD).
\end{itemize}

We  draw $I= 5000$ samples according to each sampling strategy and for each  estimator $\hat t_Y$ of $t_Y$, we compute the relative bias (RB) and the relative mean square error (RMSE): 
\begin{align*}
RB(\hat t_Y(t_d)) &= 100 \frac{ E_{MC}[\hat t_Y(t_d)] - t_Y(t_d)}{t_Y(t_d)}, \quad d=1, \ldots, D \\
RMSE(\hat t_Y(t_d)) & =100\frac{MSE_{MC}[\hat t_Y(t_d)]}{MSE_{MC}[\hat t_Y^{HT}(t_d)]},\quad d=1, \ldots, D
\end{align*}
where $E_{MC}[\hat t_Y(t_d)]=\sum_{i=1}^{I}\hat t_Y^{(i)}(t_d)/I$ and $MSE_{MC}(\hat t_Y(t_d)) = \sum_{i=1}^{I} (\hat t_Y^{(i)}(t_d) - t_Y(t_d))^2/I$ are the Monte-Carlo expectation and mean square error of $\hat t_Y(t_d)$ computed over the $I=5000$ samples and  $t_Y(t_d)$ is the real value of the total curve at   instant $t_d$. In order to assess the global performance, we consider the mean value, over time, of these indicators 
\[
RB=\frac{1}{D}\sum_{d=1}^DRB(\hat t_Y(t_d)) \quad \mbox{ and } \quad RMSE=\frac{1}{D}\sum_{d=1}^DRMSE(\hat t_Y(t_d)).
\]

\begin{table}[htbp]
\begin{center}
\small{\begin{tabular}{|c|ccc|ccc|ccc|ccc|}
   \hline
     Estimator &\multicolumn{3}{|c|}{SRS (size=)}  &\multicolumn{3}{|c|}{STR J0 (size=)} & \multicolumn{3}{|c|}{STR J10 (size=)}& \multicolumn{3}{|c|}{STR J20 (size=)} \\
  \cline{2-13}
RB (\%)&  \small{40}& \small{100} & \small{400}& \small{40}& \small{100} & \small{400} & \small{40}& \small{100} & \small{400} & \small{40}& \small{100} & \small{400} \\
  \hline
 \mbox{ minimax pointwise} & -9 & -6& -3 & -2 & -2& -1 & -4 &-3 & -1 & -5 &-4 & -2\\
 %\hline
 \mbox{4th pointwise} & -4 & -2 & -1 & -1 & -1& 0 & -2 & -1 & 0 & -2 & -1& 0\\
 %\hline
 \mbox{10th pointwise} & -7 &-4 & -2 & -2 &-2 & -1 & -3 & -2 & -1 & -4 & -3 & -1\\
% \hline
 \mbox{Robust PCA} & -7 & -5 & -2 & -1 & -1 & 0 & -3 & -2 & -1 & -3& -3 & -1\\
% \hline
 \mbox{Robust wave} & -7 & -5 & -2 & -2 & -1 & 0 & -3 & -2& -1 & -3 & -3 & -1\\
% \hline
 \mbox{MBD} & -8 & -5 & -2 & 0 & 0 & 0 & -2 & -1 & -1 & -3 & -2 & -1\\
\hline
\end{tabular}}
\end{center}
\caption{Relative bias (RB in \%).}
\label{RB}
\end{table}

\begin{table}[htbp]
\begin{center}
\small{\begin{tabular}{|c|ccc|ccc|ccc|ccc|}
   \hline
     \small{Estimator} &\multicolumn{3}{|c|}{SRS (size=)}  &\multicolumn{3}{|c|}{STR J0 (size=)} & \multicolumn{3}{|c|}{STR J10 (size=)}& \multicolumn{3}{|c|}{STR J20 (size=)} \\
  \cline{2-13}
RMSE (\%)&  \small{40}& \small{100} & \small{400} & \small{40}& \small{100} &\small{400} & \small{40}& \small{100} & \small{400} & \small{40}& \small{100} & \small{400} \\
\hline
\small{\mbox{minimax pointwise}} & 74 & 85 & 96 & 97 & 96 & 98 & 89 & 88 & 91 & 84 & 86 & 91\\
% \hline
  \small{\mbox{4th pointwise}} & 85 & 92 & 98 & 98 &97 & 99 & 93 & 94 & 97 & 91 & 93 & 97\\
% \hline
  \small{\mbox{10th pointwise}} & 77 & 86 & 96 & 97 & 96 & 98 & 90 & 89 & 92 & 85 & 88 & 93\\
% \hline
  \small{\mbox{Robust PCA}} & 73 & \textbf{83} &95 & 97 & 97 & 98 & 87 & 86 & 89 & 82 & 84 & 90\\
% \hline
  \small{\mbox{Robust wave}} & \textbf{72} & \textbf{83} & \textbf{94} & \textbf{94} & \textbf{93} & \textbf{95} & \textbf{85} & \textbf{83} &\textbf{88} &\textbf{81} & \textbf{82} &\textbf{89}\\
% \hline
 \small{ \mbox{MBD}} & 75 & 86 & 97 & 100 & 99 & 100 & 89 & 87 & 89 & 85 & 85 &  90\\
 \hline
\end{tabular}}
\end{center}
\caption{Relative MSE (RMSE in \%).}
\label{RMSE}
\end{table}

The results  are reported in Tables \ref{RB} and \ref{RMSE}.
 We can note that the use of  robust methods lead to important precision gains particularly when the sample size is small. For the SRS design and for  the best robust method, the global error is reduced by 28\%  when the sample size is 40, by 17\% when the sample size is 100. Moreover the robust methods never deteriorate significantly the global precision. The performances of these robust estimators are quite similar. Nevertheless, the functional methods based on wavelets or spherical PCA are slightly better, followed by the robust estimator built with the global truncation function based on the modified band depth.

However, robust methods tend to underestimate the population total curve because the outliers, whose influence is reduced, are often units with large values. So the robust methods lead to a negative bias of a few percents. This bias is larger for more imprecise sampling designs.

For stratified samplings without strata jumpers, the use of robust approaches do not lead to much improvement. This result is not surprising  since a good stratification permits to reduce the influence of large units during the sampling phase.  We also remark that, in this situation, the Horvitz-Thompson estimator is nearly as effective as the less accurate robust approach (MBD). We also note that the relative bias in that case is very small (less than 2\%) which could mean that the conditional biases are almost not truncated. 

 On the contrary, in presence of strata jumpers, the use of robust methods permits to improve significantly the precision, especially when the strata jumpers rate is high. The observed gains are approximately  15\% in presence of  10\% of strata jumpers. 

On this simulation study, the minimax criterion for the choice of the tuning constant gives better result than the $q$th power criterion. As expected, the performances of the robust estimators built on minimax and the $q$th power criteria are very similar for $q$ large.  
We plot in Figure \ref{temporel}, the relative mean square error along time for the suggested estimators and SRS sample of size $n=100$. We can remark that RMSE varies much over time.

\begin{figure}[!h]
\begin{center}
\includegraphics[scale = 1.05]{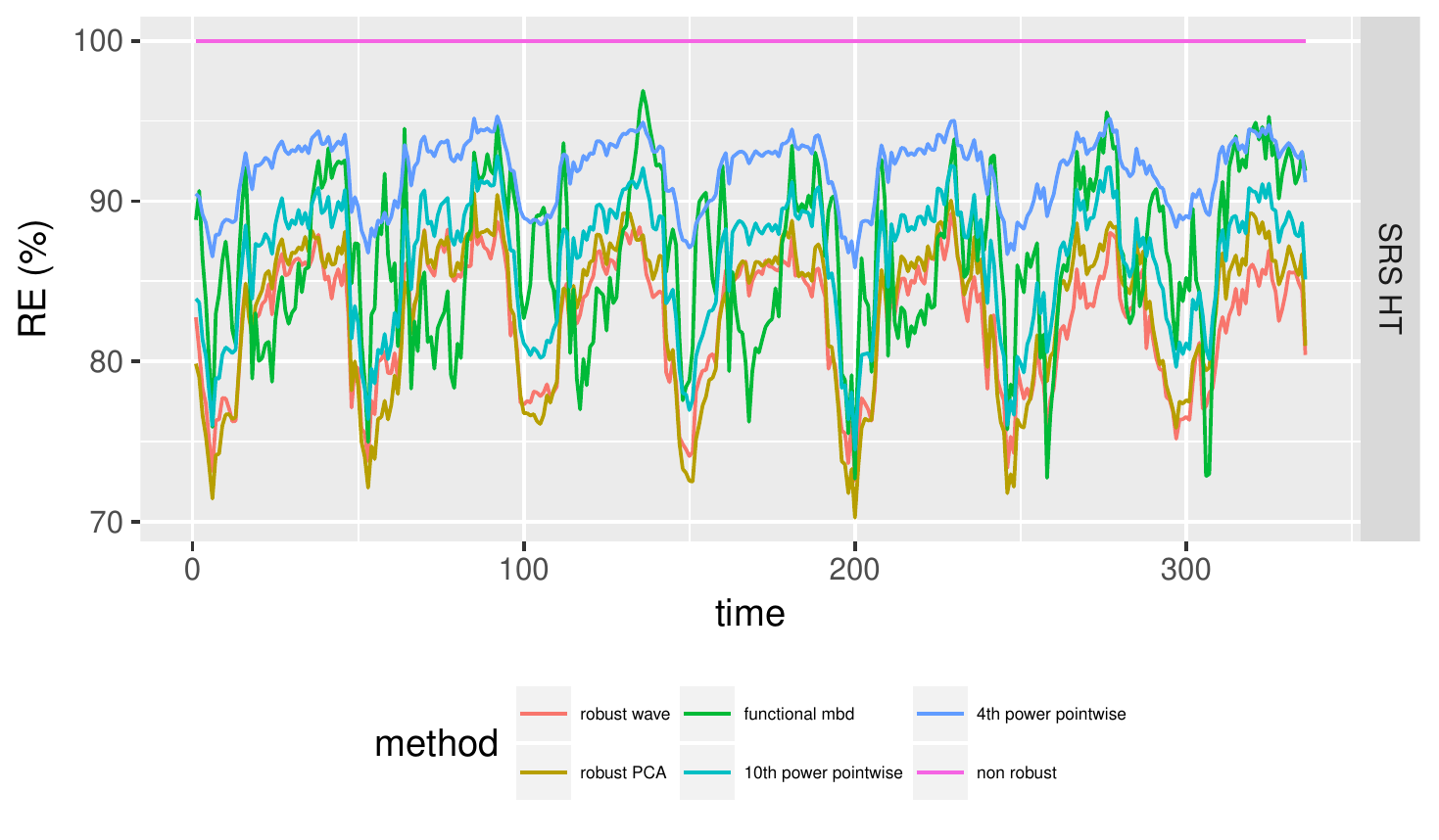} 
\caption{Evolution of the RMSE over time for different methods and SRS of size $n=100$}
\label{temporel}
\end{center}
\end{figure}

We have also computed the execution time for the suggested methods. The point-wise robust method  is the fastest robust method: for example, for $n=100,$ its mean execution time for one simulation is around $6\times 10^{-2}$ seconds. The projection methods are slightly slower but they never take more than $2\times 10^{-1}$ seconds and finally, the functional MBD estimator is around $3\times 10^{-1}$ second always for a sample of $100$ individuals. Moreover, this computation time only moderately increases when the sample size increases. 

\subsection{A comparison of the mean square error estimators}

We compare the linearization approaches with the population bootstrap and the generalized bootstrap. For the bootstrap methods, we consider $B=1000$ replications.
%We will compare the performance of the mean square error estimators obtained by linearization and bootstrap (Gross and generalized) for  robust estimators as well as for the non-robust Horvitz-Thompson estimator. %Both criteria for choosing the tuning constant are considered.  
We  compute, by means of $I=5000$ simulations, the estimators of MSE for the following total estimators: 
$\hat t_{HT}$ the usual Horvitz-Thompson estimator,    $\hat t^{(R1)}_{HT}$ the point wise estimator robustified via the minimax criterion,  $\hat t^{(R1)}_{HT}$  robustified via the $10$th 
 power criterion;  $\hat t^{(R2)}_Y$  based on robust PCA and robustified by the minimax criterion, and the total estimators based on a  wavelet expansion, $\hat t^{(R3)}_Y$, robustified via the minimax criterion. The relative bias of the estimators of  MSE  are given in Table \ref{resGlobauxVariance}. 

\begin{table}[htbp]
\begin{center}
\small{\begin{tabular}{|c|ccc|ccc|}  
   \hline
    &\multicolumn{3}{c}{SRS }  &\multicolumn{3}{c|}{STR J10} \\
%  \cline{2-7}
 &  Gen. Boot. & Gross' Boot. & Lin. &  Gen. Boot.   & Gross' Boot. & Lin. \\
 \hline
   HT & -1 & 4 & -1 & -1 & 0 & -1 \\ 
   Minimax & 23 & 4 & -27 & 20 & 0 & -17 \\
 10th power & 25 & 9 & -21 & 23 & 2 & -15 \\  
  RPCA & 24 & 5 & -32 & 22 & 2 & -26 \\  
Wavelet & 25 & 6 & -23 & 25 & 3 & -21 \\ 
\hline
\end{tabular}}
\end{center}
\caption{Relative Bias of $\widehat{MSE}$ for samples of size $n=100$, with different MSE estimation procedures.} 
\label{resGlobauxVariance}
\end{table}

We can note that, as expected, all the estimators of the mean squared error  provide reasonable results for the non robust Horvitz-Thompson estimator.  
We also note that,  in our particular context,  the MSE estimators based on linearization lead to a significant underestimation when the estimator is robust: this underestimation is about $20\%$ when the tuning constant is determined by the new criterion and about $30\%$ for the minimax criterion. This was expected because we do not take into account the variability due to the data driven selection  of the value of the tuning constant. 
On the contrary, we observe a strong overestimation of the variance of robust estimators for generalized bootstrap whereas Gross' bootstrap seems to gives satisfactory results for all the scenarios.

As far as computation time is concerned, the  MSE estimation based on linearization is quite fast, around a few tenth of second, whereas the bootstraps are significantly slower, around 20 seconds for the generalized bootstrap and 80 seconds for Gross' bootstrap.

\section{Concluding remarks }

Three types of robust estimators have been proposed in this work  in order to adapt, from the univariate to the functional case, robust estimation techniques in finite populations:
\begin{itemize}
\item Point-wise robust estimators built by truncation of the conditional bias at each instant.
\item Robust estimation based on dimension reduction methods.
\item Global functional truncation methods based on statistical depth.
\end{itemize}

These approaches have been compared on the estimation of  totals of load electricity curves. The comparisons have shown that robust methods lead to a noticeable improvement of the precision, especially when the estimation is the most imprecise (small sample sizes, sampling designs which do not include any auxiliary information or presence of very heterogeneous units in a same stratum). When the precision of the non robust estimators is already satisfying (larger sample sizes or relevant stratification), the precision gains are smaller. However, a very important fact is that robust methods never deteriorate the quality of the estimation.

We can also rank, in our simulation study, the different approaches according to their performances. The robust estimators based on wavelets expansion or on robust PCA are the most effective, followed by pointwise robust estimators and then  global functional truncation based on the notion of depth.

The corresponding mean squared errors can be estimated using linearization or bootstrap. Gross' bootstrap seems to give satisfactory results but is computationally intensive whereas linearization-based techniques are much faster but may lead to noticeable underestimations.

We have also proposed a new criterion for choosing the tuning constant based on the $q$th power of the conditional bias. Its  application on a real dataset  showed that  the minimax criterion is more effective than this new criterion.

Since our simulation studies  have shown that the use of robust methods seems to be  particularly relevant for small sample sizes,  a natural extension of the work presented here is robust estimation of curves  for small areas as considered in the PhD dissertation of \cite{anne_these}. However, robust estimation for small areas is a challenging issue. Indeed, aggregating robust small domain estimates lead to  overall estimators that may have a large bias,  as noted in \cite{rivest_hidirouglou},  \cite{Favre-Martinoz_2015} and \cite{clark_2017}. Another difficulty is the  fact that aggregated domain estimates may not be consistent with the population total estimate. To overcome this difficulty, one can use the approach suggested in  \cite{Favre-Martinoz_2015} based on a calibration technique.

\section*{Appendix}

\small{

We suppose that the sample size $n$ and the population size $N$ become large.  We consider a sequence of growing and nested populations $U_N$ with size $N$ tending to infinity and a sequence of samples $s_N$ of size $n_N$ drawn from $U_N$ according to the sampling design $p_N(s_N)$. The first and second order inclusion probabilities are respectively denoted by $\pi_{kN}$ and $\pi_{klN}$. For simplicity of notations and when there is no ambiguity, we drop the subscript $N$. To prove our asymptotic results we need to introduce the following assumptions.
\begin{itemize}
\item[\bf A1.] We assume that $\displaystyle\lim_{N\rightarrow \infty} \frac{n}{N}=\pi \in (0,1).$
\item[\bf A2.] We assume that $\displaystyle\min_{k \in U} \pi_k\geq\lambda>0$, $\displaystyle\min_{k \neq l\in U} \pi_{kl}\geq\lambda^*>0$ and

 \begin{align*}
\pi_{kl}& = \pi_k\pi_l\left\lbrace1-\frac{(1-\pi_k)(1-\pi_l)}{D(\pi)}[1+o(1)]\right\rbrace
\end{align*}
 uniformly in $k$ and $l$, where $D(\pi) = \sum_U \pi_i(1-\pi_i)$.
\item[\bf A3.] There are two positive constants  $C_2$ and $C_3$ and $\beta>1/2$ such that, for all $N$ and for all $(r,t)\in[0,T]\times [0,T]$,
\[
\frac{1}{N}\sum_{k\in U}(Y_k(0))^2<C_2 \quad \mbox{and} \quad \frac{1}{N}\sum_{k\in U}(Y_k(t)-Y_k(r))^2<C_3\vert t- r\vert^{2\beta}.
\] 
\end{itemize}

Assumptions {\bf A1} and {\bf A2} are classical hypotheses in survey sampling and deal with the first and second order inclusion probabilities. They are satisfied for high entropy sampling designs with fixed size (see for example \cite{hajek_1964}). They directly imply that $c n \leq D(\pi) \leq n,$ for some strictly positive constant $c$.  %The assumption {\bf A2} implies that  $\displaystyle\limsup_{N \rightarrow \infty}n \max_{k\neq l \in U}\vert \pi_{kl}-\pi_k\pi_l\vert<C_1<\infty$.
Assumption {\bf A3} is a  regularity condition on the individual trajectories. Even if point-wise consistency, for each fixed value of $t$, can be proven without any condition on $\beta$,  this regularity condition is required to get the uniform convergence of the mean estimator (see \cite{cardot2011horvitz}). 

\begin{Prop}
Suppose that {\bf A1} and {\bf A3}  are fulfilled and the sampling design is simple random sampling without replacement or suppose that hypotheses {\bf A1}-{\bf A3} are fulfilled. Then 
\[
\sup_{t \in [0,T]} \left| \sum_{i\in s}d_iA_i(t)-\sum_{i\in U}A_i(t) \right| =O_p(n^{-1/2}).
\]
\end{Prop}

\begin{proof}
Recall that, for $t \in [0,T]$,
\begin{align*}
A_i(t) &=\frac{-1}{1-\pi_i}\sum_{j\in U, j\neq i}\frac{\pi_{ij}-\pi_i\pi_j}{\pi_j}Y_j(t). 
\end{align*}

For simple random sampling without replacement, $\pi_i = n/N$ and $\pi_{ij} = n(n-1)/(N(N-1))$ for $i\neq j$, and we have that (with $d_i = 1/\pi_i$),
\begin{align}
\sum_{i\in s}d_iA_i(t)-\sum_{i\in U}A_i(t) & = \frac{1}{N-1} \left( t_Y(t)  - \widehat{t}_Y(t) \right)
\end{align}
The result is then a direct consequence of Proposition 3.1 in \cite{cardot2011horvitz}. 

\medskip

Consider now the more general case of fixed-size high entropy sampling designs. 
Introducing the approximation to the second order inclusion probabilities in $A_i$ we get after some algebra
\begin{align}
\sum_{i\in s}d_iA_i(t)-\sum_{i\in U}A_i(t) \approx \frac{1}{D(\pi)} \left( \sum_U \pi_i (1-\pi_i) Y_i(t) -   \sum_s d_i \pi_i (1-\pi_i) Y_i(t) \right), \quad t \in [0,T]. 
\end{align}
The weighted trajectories  $\pi_i(1-\pi_i)Y_i(t), t \in [0,T]$ also satisfy assumption {\bf A3} and the  result is a consequence of Proposition 3.1 in \cite{cardot2011horvitz} (see also \cite{cardot_goga_lardin_scandin}). 
\end{proof}
}

\end{document}